\documentclass[preprint]{elsarticle}%../IEEEtran

\usepackage[utf8]{inputenc}
\usepackage[english]{babel}
\usepackage{csquotes}
\usepackage{amsmath, amssymb, amsfonts, amsthm}

\usepackage{braket}
\usepackage{graphicx}

\usepackage{xcolor}
\definecolor{URLColorBlue}{HTML}{2A1B81}
\definecolor{BlueGreen}{cmyk}{0.85,0,0.33,0}
\definecolor{RawSienna}{HTML}{A0522D}
\definecolor{TealLight}{HTML}{00C798}
\definecolor{Teal}{HTML}{008080}
\definecolor{orange}{HTML}{FF7700}
\definecolor{enji-ro}{HTML}{9D2933}
\definecolor{hana'asagi}{HTML}{1D697C}
\definecolor{seiheki}{HTML}{3A6960}
\definecolor{too'oo}{HTML}{FFB61E}
\definecolor{byakuroku}{HTML}{A5BA93}
\definecolor{konjoo-iro}{HTML}{003171}
\definecolor{hanada}{HTML}{044F67}
\definecolor{koorozen}{HTML}{592B1F}
\definecolor{ootan}{HTML}{FF4E20}
\definecolor{arazome}{HTML}{FFB3A7}
\definecolor{midori}{HTML}{2A603B}

\usepackage{hyperref}

\hypersetup{colorlinks,
			citecolor=TealLight,
			urlcolor=URLColorBlue,
			linkcolor=hana'asagi,
			pdftitle={Code generator matrices as RNG conditioners},
			pdfauthor={Alessandro Tomasi, Alessio Meneghetti, Massimiliano Sala},
			pdfkeywords={Random number generation} {Weight distribution} {Entropy extractor} {Cryptography},
}

\theoremstyle{plain}% default
\newtheorem{theorem}{Theorem}
\newtheorem{lemma}{Lemma}

\newtheorem{corollary}{Corollary}

\theoremstyle{definition}
\newtheorem{definition}{Definition}

\theoremstyle{remark}

\begin{document}

\begin{frontmatter}

\title{Code generator matrices as RNG conditioners}
% \author{\IEEEauthorblockN{A. Meneghetti
%\IEEEauthorrefmark{1}
% ,
% M. Sala and A. Tomasi}

% \IEEEauthorblockA{Department of Mathematics,
% University of Trento}}
%\author{A. Tomasi \and A. Meneghetti \and M. Sala}
%\affil[1]{Department of Mathematics, University of Trento}
\author[unitn]{A.~Tomasi\corref{cor0}}
\ead{twin.ion.engine@gmail.com}
\author[unitn]{A.~Meneghetti\corref{cor1}}
\ead{almenegh@gmail.com}
\author[unitn]{M.~Sala\corref{cor1}}
\ead{maxsalacodes@gmail.com}

\cortext[cor0]{Principal corresponding author}
\cortext[cor1]{Corresponding author}
\address[unitn]{Department of Mathematics, University of Trento, Via Sommarive 14, 38123 Povo, IT}

\begin{abstract}
We quantify precisely the distribution of the output of a binary random number generator (RNG) after conditioning with a binary linear code generator matrix by showing the connection between the Walsh spectrum of the resulting random variable and the weight distribution of the code. Previously known bounds on the performance of linear binary codes as entropy extractors can be derived by considering generator matrices as a selector of a subset of that spectrum. We also extend this framework to the case of non-binary codes.
\end{abstract}


\end{frontmatter}

%\maketitle

%\abstract{We quantify precisely the distribution of the output of a binary random number generator (RNG) after conditioning with a binary linear code generator matrix by showing the connection between the Walsh spectrum of the resulting random variable and the weight distribution of the code. Previously known bounds on the performance of linear binary codes as entropy extractors can be derived by considering generator matrices as a selector of a subset of that spectrum. We also extend this framework to the case of non-binary codes.}

\section{Introduction}

Our objective is to precisely quantify the result of applying any one specified code generator matrix a single time as conditioning function to the output of an entropy source. We follow the recommendations set out by NIST \cite{NIST_800-90B} for the precise meaning to be given to these terms.

Linear transformations based on codes have previously been applied to sources of entropy producing output that can be treated as independent but biased bit sequences; bounds on the statistical distance of such output from the uniform distribution have been shown in \cite{Lacharme_08, Lacharme_09, Zhou_Bruck_11_extractors}. The application of these functions is generally presented as part of the framework of randomness extractors, as summarised for instance in \cite{Shaltiel_11}, in the sense that random matrices are chosen with specific properties, such as minimum distance of the code, or an approximate distribution of the weights. The performance of these functions is usually presented in the form of bounds on the statistical (total variation) distance of the resulting conditioned output from the uniform distribution. The present work extends and complements the known results by quantifying precisely the statistical distribution of the output after conditioning with a specified generator matrix by showing the connection between the probability mass function of the resulting random variable and the weight distribution of the code; the known bounds can then be derived as special cases.

We treat binary streams in groups of $k$ bits as discrete random variables $X$, in the sense that the number of possible outcomes is finite and the variables admit a discrete probability mass function $\mu_X(j) = \mathbb{P}(X=x_j)$; moreover, we begin by considering these variables to take values in a finite field $\mathbb{F}_p$ or a vector space $\left(\mathbb{F}_p\right)^k$, with particular regard to the special case of binary variables, $p=2$. In Section \ref{sec_Walsh_TVD} we show the connection between the Walsh spectrum of $X$ and the bias of individual bits $X(j)$, and use this in Section \ref{sec_Walsh_binary_extractor} to show how previously known bounds can be derived by considering generator matrices as a selector of a subset of that spectrum. We then extend this framework to the case of output in non-binary finite fields by use of the Fourier transform in Section \ref{sec_Fourier_extractor}.

\section{Total variation distance and the Walsh-Hadamard transform}
\label{sec_Walsh_TVD}

We show in the following one way in which the Walsh-Hadamard transform may be used to bound the total variation distance of binary random variables with a known probability mass function. This may seem an unnecessary exercise since the $\mathrm{TVD}$ can simply be computed exactly from this knowledge, but aside from revealing some interesting structure to the calculation it will become more explicitly useful in the following section.

Consider a random vector $Y \in (\mathbb{F}_2)^k$ with probability mass function
\begin{align*}
	\mu_Y		& \in \mathbb{R}^{2^k}, 	\\
	\mu_Y (j)	& = \mathbb{P}(Y = \mathbf{j})
\end{align*}
where in writing $j$ and $\mathbf{j}$ we use the binary representation of integers as vectors $a \in \mathbb{Z}_{2^k}$
\begin{align*}
	\mathbf{a}	& = \left\{a_j \,\middle|\, a = \sum_{j=0}^{k-1} a_j 2^j \right\} \in (\mathbb{F}_2)^k	\,.
\end{align*}
The $a$-th order Walsh function evaluated at $b$ is
\begin{align}	\label{eqn_Walsh_fn_order_a}
	h_a(b)	& = (-1)^{ \mathbf{a} \cdot \mathbf{b} }
\end{align}
with $\cdot$ the dot product on $(\mathbb{F}_2)^k$. The $j^{th}$ Walsh characteristic function of Y as defined in \cite{Pearl_71} is
\begin{align}
	\chi_j(Y)	& = \sum_{a=0}^{2^k-1} h_j (a) \mu_Y(a)	\\
			& = \mathbb{E}[h_j(Y)]			\label{eqn_def_chi_j_Walsh_EV}
\end{align}
Note that the dot product of two binary vectors $\mathbf{b} \cdot \mathbf{v}$ is the bitwise sum, i.e.\ the linear combination, of those elements $\mathbf{v}(i)$ that correspond to the non-zero entries of $\mathbf{b}$; therefore, the random variable
\begin{align*}
	h_b(Y)	& = (-1)^{ \mathbf{b} \cdot Y }
\end{align*}
will take value $-1$ if the linear combination of the selected elements of $Y$ is equal to one, and $1$ otherwise. %The selected subset can be written as
%
%\begin{align*}
%	\set{\mathbf{b}\cdot Y}	& = \left\{ Y_i \mid b_i = 1 \right\}.
%\end{align*}
The sum of the selected elements is itself a random variable, $B = \mathbf{b} \cdot Y$ following the Bernoulli distribution with probability $\mu_B(1)$ of being equal to $1$; it follows that
\begin{align*}
	h_b(Y)	& = 1-2B,
\end{align*}
and hence we can conclude that
\begin{align*}
	\chi_b(Y)	& = \mathbb{E}[h_b(Y)]	\\
			& = 1 - 2 \mu_B(1).
\end{align*}

We note now that the bias of a Bernoulli random variable $B \in \mathbb{F}_2$ is commonly defined as
\begin{align*}
	\frac{\varepsilon_B}{2}	& = \frac{1}{2} \left| \mathbb{P}(B=1) - \mathbb{P}(B=0) \right|	\\
				& = \frac{1}{2} \left| 2 \mathbb{P}(B=1) - 1 \right|	\\
				& = \frac{1}{2} \left| 2 \mathbb{E}[B] - 1 \right|,
\end{align*}
and we observe that the Walsh characteristic of $\mathbf{b}\cdot Y$ leads to the bias of the $b^{th}$ linear combination of the elements of $Y$ via the relation
\begin{align}	\label{eqn_chi_and_bias_Y}
	\left| \chi_b(Y) \right|	& = \varepsilon_{\mathbf{b}\cdot Y}	\,.
\end{align}
In particular, the combinations corresponding to exact powers of two, $b=2^j$, lead to the bias of each individual element of $Y$; and the zeroth Walsh characteristic, corresponding to $b=0$, will be equal to $1$ in all entries, in all cases.

%By way of example, if $Y \in (\mathbb{F}_2)^2$ is i.i.d.\, with each bit therefore having equal $\varepsilon_y$, then
%
%\begin{align*}
%	\set{\varepsilon_{\mathbf{b}\cdot Y}}_{b=0}^3	& = \set{1, \varepsilon_y, \varepsilon_y, \varepsilon_y^2}	\,.
%\end{align*}

The set $\set{h_i}$ is known to correspond to the rows of a Hadamard matrix $H$ of size $2^k$; the set of all Walsh characteristics of $Y$ can thus be written compactly in matrix notation as
\begin{align}	\label{eqn_chi_def}
	\chi(Y)	& = H\mu_Y	\,.
\end{align}
As a matter of notation, for a uniformly distributed random variable $U \in \left(\mathbb{F}_2\right)^k$ we have
\begin{align*}
	\mu_U	& = \frac{\mathbf{1}}{2^k}	\\
	\chi(U)	& = \mathbb{I}_{2^k}(\cdot, 1)
\end{align*}
with $\mathbf{1}$ a column vector of ones and $\mathbb{I}_{2^k}(\cdot, 1)$ the first column of the identity matrix of size $2^k$. We may use this to estimate the total variation distance of $Y$ from uniform as follows.
%
%This can equivalently be expressed in terms of the total variation distance of the probability mass function $p$ from the uniform over $\mathbb{F}_2$, as well as the expectation of $B$:
%
%\begin{align*}
%	\mathrm{TVD}(B, \mathcal{U})	& = \frac{1}{2} \left\| p - \frac{\mathbf{1}}{2} \right\|_1	\\
%				& = \frac{1}{2} \left( \left| p_1 - \frac{1}{2} \right| + \left| p_0 - \frac{1}{2} \right| \right)	\\
%				& = \frac{1}{2} \left( \left| p_1 - \frac{1}{2} \right| + \left| \frac{1}{2} - p_1 \right| \right)	\\
%				& = \frac{1}{2} \left| 2 p_1 - 1 \right|	\\
%				& = \frac{1}{2} \left| p_1 - (1-p_1) \right|	\\
%				& = \frac{\varepsilon_y}{2}
%\end{align*}
%
%\begin{align*}
%	\frac{\varepsilon_y}{2}	& = \frac{1}{2} \left| 2 p_1 - 1 \right|	\\
%				& = \frac{1}{2} \left| 2 \mathbb{E}[B] - 1 \right|
%\end{align*}
%
\begin{theorem}	\label{thm_Hadamard_bound_Y}

	The total variation distance of a random vector $Y \in (\mathbb{F}_2)^k$ from uniform $U$,
	\begin{align*}
		\mathrm{TVD}(Y, \mathcal{U})	& = \frac{1}{2} \left\| \mu_Y - \mu_U \right\|_1	\\
					& = \frac{1}{2} \delta_Y
	\end{align*}
	is bounded by the sum of the bias of all non-trivial linear combinations of the output bits,
	\begin{align*}
		\delta_Y	& \leq \sum_{\mathbf{b} \in (\mathbb{F}_2)^k \setminus \mathbf{0}} \varepsilon_{\mathbf{b}\cdot Y}	\,.
	\end{align*}

\end{theorem}
\begin{proof}
	\begin{align}
		\left\| \mu_Y - \mu_U \right\|_2	& = \left\| \frac{H^T H}{2^k} \left(\mu_Y - \mu_U \right) \right\|_2	\label{eqn_tvd_step_02}	\\
					& = \frac{1}{2^{k/2}} \left\| \frac{H^T}{2^{k/2}} \left(\chi(Y) - \chi(U) \right) \right\|_2	\label{eqn_tvd_step_03}	\\
					& \leq \left\| \chi(Y) - \chi(U) \right\|_1	\label{eqn_tvd_step_04}	\\
					& = 	\sum_{\mathbf{b} \in (\mathbb{F}_2)^k \setminus \mathbf{0}} \varepsilon_{\mathbf{b}\cdot Y}	\label{eqn_tvd_step_05}
	\end{align}
	Here $H^T$ is the transpose of the Hadamard matrix $H$. Equation (\ref{eqn_tvd_step_02}) follows from $H^T/\sqrt{2^k}$ being the unitary inverse Hadamard transform; Equation (\ref{eqn_tvd_step_03}) uses the definition of $\chi(Y)$ in Equation (\ref{eqn_chi_def}); and lastly, the bound in Equation (\ref{eqn_tvd_step_04}) stems from the $\ell_1$ bound on $q$-dimensional vector spaces, $\|\cdot\|_1 \leq q^{1/2}\|\cdot\|_2$.
\end{proof}
\begin{corollary}	\label{crl_Hadamard_bound_Y}
	
	If the bits $Y(j) \in \mathbb{F}_2$ are i.i.d.\ with known bias $\varepsilon_y$, then
	\begin{align*}
		\delta_Y	& \leq \sum_{l=1}^k A_l \varepsilon_y^l	\\
				& = \sum_{l=1}^k \binom{k}{l} \varepsilon_y^l
	\end{align*}
	where $A_l$ is the number of $\mathbf{b} \in (\mathbb{F}_2)^k$ with Hamming weight $w(\mathbf{b}) = l$.

\end{corollary}
%\begin{proof}

%	Continuing from Equation (\ref{eqn_tvd_step_05}),
	%
%	\begin{align*}
%		\delta_Y	& \leq \sum_{l=1}^k \left| \mathbb{P}\left(\sum_1^l y = 1\right) - \mathbb{P}\left(\sum_1^l y = 0\right) \right|
%	\end{align*}
%\end{proof}

\section{W-H bound on binary generator matrices as extractors}
\label{sec_Walsh_binary_extractor}

We now consider the previous bound as applied to random variables $Y = G X$, with $X \in (\mathbb{F}_2)^n$ a sequence of $n$ Bernoulli random variables with known probability mass $\mathbb{P}(X=b) = \mu_X(b)$ and identical bias $\varepsilon_X = |1-2\mathbb{P}(X(i) = 1)|$ for each bit, and $G \in (\mathbb{F}_2)^{k \times n}$ the generator matrix of an $(n, k, d)$ linear code $C$ with weight distribution $\set{A_l}$; in other words, $C$ is a subspace of $(\mathbb{F}_2)^n$ and $A_l$ is the number of $\mathbf{c} \in C$ with Hamming weight $w(\mathbf{c}) = l$.

The definition of Walsh characteristic functions as expected values in Equation (\ref{eqn_def_chi_j_Walsh_EV}) directly leads to
\begin{align*}
	\chi_b(GX)	& = \mathbb{E}[h_b(GX)]	\\
			& = \sum_{x=0}^{2^n-1} (-1)^{\mathbf{b}\cdot G\mathbf{x}} \mu_X(x)
\end{align*}
We note here that the dot product in the Walsh function can equivalently be expressed using the transpose $\mathbf{b}^T$ as
\begin{align*}
	\mathbf{b}\cdot G\mathbf{x}	& = \mathbf{b}^T G\mathbf{x}	\,,
\end{align*}
and in particular the product
\begin{align*}
	\mathbf{c}^T = \mathbf{b}^T G
\end{align*}
is a linear combination of the rows of $G$: since $G$ is the generator matrix of a linear code $C$ the rows of $G$ form a basis of $C$, and hence any linear combination of them is again a word $\mathbf{c} \in C$. Consequently, just as the Walsh characteristic led to a measure of bias in Equation (\ref{eqn_chi_and_bias_Y}), we can conclude that
\begin{align}	\label{eqn_chi_and_bias_GX}
	\left| \chi_b(Y) \right|	& = \varepsilon_{\mathbf{c}\cdot X}	\,.
\end{align}
In other words, the bias of the $b^{th}$ element of $Y$ is equal to the bias of a linear combination of $w(\mathbf{c})$ bits of $X$ (compare to Equation (\ref{eqn_chi_and_bias_Y})). This leads directly to the following bound.

\begin{theorem}	\label{thm_Hadamard_bound_GX}

	Let $Y = G X$, with $X \in (\mathbb{F}_2)^n$ a sequence of $n$ independent but not necessarily identically distributed Bernoulli random variables, and $G \in (\mathbb{F}_2)^{k \times n}$ the generator matrix of an $(n, k, d)$ linear code $C$. The total variation distance of the random variable $Y \in (\mathbb{F}_2)^k$ from uniform,
	\begin{align*}
		\mathrm{TVD}(Y, \mathcal{U})	& = \frac{\delta_Y}{2}
	\end{align*}
	is bounded by the sum of the bias of all linear combinations of the bits in $X$ defined by the codewords of $C$, in the following measure:
	\begin{align*}
		\delta_Y	& \leq \sum_{\mathbf{b} \in (\mathbb{F}_2)^k \setminus \mathbf{0}} \varepsilon_{\mathbf{b}^TG\cdot X}	\\
				& = \sum_{\mathbf{c} \in C \setminus \mathbf{0}} \varepsilon_{\mathbf{c} \cdot X}
	\end{align*}

\end{theorem}
\begin{corollary}	\label{crl_Hadamard_bound_GX}
	
	If the bits $X(j) \in \mathbb{F}_2$ are i.i.d.\ with known bias $\varepsilon_x$, then
	\begin{align*}
		\delta_Y	& \leq \sum_{l=d}^n A_l \varepsilon_{x}^l	\,.
	\end{align*}
	with $\set{A_l}$ the weight distribution of $C$.

\end{corollary}

Note that Corollary \ref{crl_Hadamard_bound_Y} is closely related to Corollary \ref{crl_Hadamard_bound_GX} if we consider that in this context the $\set{A_l}$ in the former correspond precisely to the weight distribution of the trivial code given by the message space itself, $\left(\mathbb{F}_2\right)^k$. A particular case of Corollary \ref{crl_Hadamard_bound_GX} for strictly binomial $\set{A_l}$ was proved in \cite{Zhou_Bruck_11_extractors}, Theorem 6. We can thus recover the following known bound (see \cite{Lacharme_08}, Theorem 1):

\begin{corollary}	\label{cor_Lacharme}

	Considering only the minimum distance $d$ rather than the full weight distribution,
	\begin{align*}
		\delta_Y	& \leq  (2^k-1) \varepsilon_x^d	\,.
	\end{align*}

\end{corollary}

\section{Total variation distance and the Fourier transform}
\label{sec_Fourier_TVD}

The Hadamard transform is a special case of the Fourier transform constructed with primitive $2$-nd root of unity $p=2$, $\omega_p = -1$, and the Hadamard matrix of size $2^k$ is constructed by the Kronecker product $H_{2^k} = H_2 \otimes H_{2^{k-1}}$, so the binary case considered in Section \ref{sec_Walsh_binary_extractor} can be seen as a special case. Employing the Fourier transform is natural in this setting and closely follows well-established techniques for the sum of continuous random variables, which have their own convolution theorem and proofs of convergence to a limiting distribution.

%The characteristic function of a continuous random variable $Y\in \mathbb{R}$ is commonly defined as
%
%\begin{align*}
%	\chi_Y(t)	& = \mathbb{E}\left[e^{itY}\right]	\qquad i = \sqrt{-1}\,;
%\end{align*}
%
%this function is known to always exist for any integrable $Y$ and to be closely related to the Fourier transform by a simple change of variable, provided $Y$ admits a density function $\mu_Y$:
%
%\begin{align*}
%	\chi_Y(-2\pi t)	& = \mathbb{E}\left[e^{-2\pi itY}\right]	\\
%			& = \int_{\Omega_Y}e^{-2\pi itx} \mu_Y(x) dx	\\
%			& = \mathcal{F}(\mu_Y(t))
%\end{align*}
%
%For any real-valued $Y$ satisfying these constraints, the Fourier transform is known to always exist and be uniformly continuous, allows the unique recovery of the original distribution via the inverse transform, and is often used to describe limiting behaviour of distributions; for instance, of particular interest here is that given a set of $n$ independent random variables $\set{X_j}$ and constants $\set{a_j}$, the characteristic function of the sum is the product of the characteristic functions of each $X_j$:
%
%\begin{align*}
%	S_n		& = \sum_{j=1}^n a_j X_j	\\
%	\chi_{S_n}(t)	& = \prod_{j=1}^n \chi_{X_j}(a_j t)
%\end{align*}
%
%This is simply a restatement of the well-known convolution theorem, which states that the distribution of the sum of two independent random variables is the convolution of the individual distributions.

Given an integer $a \in \mathbb{Z}_{p^k}$, we denote its $p$-ary representation by
\begin{align*}
	\mathbf{a}	& = \left\{a_j \,\middle|\, a = \sum_{j=0}^{k-1} a_j p^j \right\} \in \left( \mathbb{F}_p \right)^k	\,.
\end{align*}
We shall use this notation interchangeably in the following as a natural indexing of the elements of $\left( \mathbb{F}_p \right)^k$.

Consider a random variable $Z \in \mathbb{F}_p$ with probability mass function
\begin{align*}
	\mu_Z		& \in \mathbb{R}^{p}, 	\\
	\mu_Z(j)	& = \mathbb{P}(Z = \mathbf{j})	\,.
\end{align*}
Note that this implicitly assumes an ordering of the mass function $\mu_Z$ corresponding to the representation of elements $\beta_j \in \mathbb{F}_p$ as integers. The discrete Fourier transform of $\mu_Z$ may then be written in matrix form as
\begin{align*}
	F_p\mu_Z = \lambda_Z	\,,
\end{align*}
where $\lambda$ is the set of eigenvalues of the circulant matrix generated by $\mu_Z$. Indeed, the above can be restated in terms of the unitary DFT,
\begin{align*}
	\hat{F}_p	= \frac{F_p}{\sqrt{p}}	\qquad
	\hat{F}_p^*	= \frac{F_p^*}{\sqrt{p}}
\end{align*}
with $F_p^*$ the conjugate transpose of $F_p$, diagonalising the circulant matrix $C_Z$ generated by $\mu_Z$:
\begin{align*}
	C_Z				& = circ(\mu_Z)	\\
	\hat{F}_p C_Z \hat{F}_p^*	& = \Lambda_Z
\end{align*}
with $\Lambda_Z$ the $p \times p$ diagonal matrix containing all eigenvalues of $C_Z$. Note that for a uniformly distributed random variable $U \in \mathbb{F}_p$, we have
\begin{align*}
	\mu_U		& = \frac{\mathbf{1}}{p}	\\
	\lambda_U	& = F_p \mu_U = \mathbb{I}_p(\cdot, 1)
\end{align*}
where $\mathbf{1}$ is a vector of ones of length $p$, and the only non-zero eigenvalue is the zeroth one, so the full set $\lambda_U$ corresponds to the first column of the identity matrix of size $p$.

\begin{lemma}	\label{lemma_ell2_Z}

	The mass function $\mu_Z$ of a random variable $Z \in \mathbb{F}_p$ satisfies
	\begin{align*}
		\left\| \mu_Z - \mu_U \right\|_2	& = \frac{1}{\sqrt{p}} \| \lambda_Z - \lambda_U \|_2
	\end{align*}
	where $\lambda_Z = F_p\mu_Z$ is the discrete Fourier transform of $Z$.

\end{lemma}
\begin{proof}
	\begin{align}
		\left\| \mu_Z - \mu_U \right\|_2	& = \left\| \frac{\hat{F}^*}{\sqrt{p}} \left(\lambda_Z - \lambda_U \right) \right\|_2	\label{eqn_ell2_Z_step_01}	\\
						%& = \frac{1}{\sqrt{p}} \left\| \hat{F}^* \left(\lambda - \mathbb{I}(:,1) \right) \right\|_2	\label{eqn_ell2_Z_step_02}	\\
						& = \frac{1}{\sqrt{p}} \left\| \lambda_Z - \lambda_U \right\|_2	\label{eqn_ell2_Z_step_03}	\\
						& = \frac{1}{\sqrt{p}} \left( \sum_{j=1}^{p-1} \left( \lambda_Z(j) \right)^2 \right)^{1/2}	\label{eqn_ell2_Z_step_04}	
%					& = 	\sum_{\mathbf{b} \in \mathbb{F}_p \setminus 0} \lambda_{\mathbf{b}\cdot Y}	\label{eqn_tvd_step_05}
	\end{align}
	Equation (\ref{eqn_ell2_Z_step_03}) follows from the unitary Fourier transform preserving $\ell_2$ distance.%, and $\mathbb{I}(:,1)$ is the first column of the identity matrix of size $p$.
%	$\mathbf{1}$ and $\mathbf{i}$ are as defined in Equation (\ref{eqn_io}); $F^T$ is the transpose of the Fourier matrix $F$, and Equation (\ref{eqn_tvd_step_02}) follows using Equation (\ref{eqn_H_product_identity}); Equation (\ref{eqn_tvd_step_03}) uses the definition of $\chi(Y)$ in Equation (\ref{eqn_chi_def}) and the Hadamard property in (\ref{eqn_H_symmetric}); and lastly, the bound in Equation (\ref{eqn_tvd_step_04}) stems from the $\ell_1$ bound (\ref{eqn_Hadamard_ell1_bound}).
%
\end{proof}

We can obtain a first, crude bound on the $\mathrm{TVD}$ by considering the largest non-trivial eigenvalue, defined as follows for future reference.

\begin{definition}	\label{defn_lambda_star}

	Given a random variable $Z \in \mathbb{F}_p$ with mass function $\mu_Z$ and eigenvalues $F_p \mu_Z = \lambda_Z$, denote the greatest non-trivial eigenvalue by
	\begin{align*}
		\lambda_{Z*} = \max_{1\leq j\leq p-1} |\lambda_Z(j)|
	\end{align*}

\end{definition}

\begin{theorem}	\label{thm_Fourier_bound_Z}

	The total variation distance of a random variable $Z \in \mathbb{F}_p$ from uniform,
	\begin{align*}
		\mathrm{TVD}(Z, \mathcal{U})	& = \frac{\delta_Z}{2}
	\end{align*}
	is bounded by
	\begin{align*}
		\delta_Z \leq \left(p-1\right)^{1/2} \lambda_{Z*}
	\end{align*}
	with $\lambda_{Z*}$ as in Definition \ref{defn_lambda_star}.
	
\end{theorem}
\begin{proof}

	This follows from considering the worst-case scenario in which all eigenvalues $\lambda_Z$ in Lemma \ref{lemma_ell2_Z}, except the zeroth eigenvalue $\lambda_Z(0)=1$, are equal to the greatest non-trivial eigenvalue $\lambda_{Z*}$ by applying the bound on $p$-dimensional vector spaces $\|x\|_1\leq p^{1/2}\|x\|_2$.

\end{proof}

We can now consider how this affects the distribution of a sum of two variables, $S_2 = X_0 + X_1 \in \mathbb{F}_p$, which is the discrete convolution of the two probability masses,
\begin{align*}
	\mathbb{P}(S_2 = r)	& = \sum_{j \in \mathbb{F}_p} \mathbb{P}(Z_0 = j)\mathbb{P}( Z_1 = r-j)	\\
				& = \sum_{j=0}^{p-1} \mu_{Z_1}(r-j) \mu_{Z_0}(j)
\end{align*}
The distribution of $S_2$ may then be expressed in matrix notation as
\begin{align}	\label{eqn_mu_S_2}
	\mu_{S_2} = C_{Z_1} \mu_{Z_0}	\,,
\end{align}
where $C_{Z_1}$ is the circulant matrix uniquely defined by $\mu_{Z_1}$. In other words, the entry $r, j$ of $C_{Z_1}$ contains the measure under $Z_1$ of the element $\beta_r-\beta_j \in \mathbb{F}_p$ % such that $\beta_{r-j} + \beta_j = \beta_r$, which is to say that $\beta_{r-j}$ is the sum of $\beta_r$ and the additive inverse of $\beta_j$. We denote this in matrix form by
, which we denote in matrix form by
\begin{align*}
	C_{Z_1}	& = \mu_{Z_1}(B)	\,,\\
	B(r,j)	& = \beta_r - \beta_j
\end{align*}
%
%The circulant structure of $B$, and hence of $C$, can be viewed as a result of the sum over $\mathbb{F}_p$ being circular modulo $p$.

Considering the particular case of summing two identical variables $Z$ with probability mass function $\mu_Z$, the distribution of $S_2 = Z+Z$ may be written as
\begin{align*}
	S_2 \sim C_Z \mu_Z,
\end{align*}
where $C_Z$ is the circulant matrix defined uniquely by $\mu_Z$ itself. By induction,
\begin{align}
	S_n	& \sim C_Z^{n-1}\mu_Z		\nonumber \\
		& \sim \left(\prod_{j=1}^{n-1} \frac{F_p^*}{p} \Lambda_Z F_p \right)\mu_Z 	\nonumber \\
		& \sim \frac{\hat{F}_p^*}{\sqrt{p}} \Lambda_Z^{n-1} F_p\mu_Z	\nonumber 	\\
		& \sim \frac{1}{\sqrt{p}} \hat{F}_p^*\lambda_Z^n	\label{eqn_mu_S_n}
\end{align}
As well as being conceptually equivalent to using the convolution theorem, this may also be seen as considering $S_n$ as a Markov chain
\begin{align*}
	S_0	& = Z	\\
	S_j	& = S_{j-1}+Z
\end{align*}
with transition matrix $C_Z$.

Lemma \ref{lemma_ell2_Z} may be extended as follows.
\begin{lemma}	\label{lemma_ell2_S_n}

	The probability mass function $\mu_{S_n}$ of a random variable $S_n = \sum_{j=1}^n Z, \, Z \in \mathbb{F}_p$ satisfies
	\begin{align*}
		\left\| \mu_{S_n} - \mu_U \right\|_2	& = \left(\frac{p-1}{p}\right)^{1/2} \| \left( \lambda_{S_n} - \lambda_U \right)^n \|_2
	\end{align*}

\end{lemma}
\begin{proof}

	The proof is substantially the same as that of Lemma \ref{lemma_ell2_Z}, using Equation (\ref{eqn_mu_S_n}).

\end{proof}
\begin{lemma}	\label{lemma_Fourier_bound_S_n}
	
	The total variation distance of $S_n = \sum_{j=1}^n Z, \, Z \in \mathbb{F}_p$ from uniform may be bounded by
	\begin{align}
		\delta_{S_n}	& \leq \left(p-1\right)^{1/2} \lambda_{Z*}^n
	\end{align}
	with $\lambda_{Z*}$ as in Definition \ref{defn_lambda_star}.

\end{lemma}
\begin{proof}

	The proof follows by applying Lemma \ref{lemma_ell2_S_n} in the same way as Lemma \ref{lemma_ell2_Z} was applied to Theorem \ref{thm_Fourier_bound_Z}, i.e.\ assuming each of the $p-1$ eigenvalues in Lemma \ref{lemma_ell2_S_n} that are of magnitude less than $1$ to be bounded by $\lambda_*^n$, using Equation (\ref{eqn_mu_S_n}).

\end{proof}
Lemma \ref{lemma_Fourier_bound_S_n} is a slight improvement on a known bound on the convergence rates of Markov chains on Abelian groups; see e.g. \cite{Rosenthal_95} Fact 7.

We have so far assumed an ordering of the mass function $\mu_X$ of a random variable $X \in \mathbb{F}_p$ according to the representation of the elements of $\mathbb{F}_p$ as integers. Similarly for vector spaces $X \in \left( \mathbb{F}_p \right)^k$ we may assume an ordering of $\mu_X$ by least significant digit. %, so that the entries of $\mu_X$ are ordered as $[0 < 1 < \ldots < p-1 < \alpha < \alpha+1 < \ldots < 2\alpha < 2\alpha+1 \ldots]$, with $\alpha$ the generator of the multiplicative subgroup of $\left( \mathbb{F}_p \right)^k$.
Generalising to the distribution of the sum $S_2$ of two random variables $X_0, X_1 \in (\mathbb{F}_p)^k$, this may still be expressed in a form such as Equation (\ref{eqn_mu_S_2}), but the matrix $C_{X_1}$ is a level $k$ block circulant with circulant base blocks of size $(p \times p)$. Concretely, while considering all coefficients of $B$ as elements of $(\mathbb{F}_p)^k$, we may write
\begin{align*}
	B_p		& = circ( [\mathbf{0}, \mathbf{1}, \ldots \mathbf{p-1}] )		\\
	B_p^{\circ 2}	& = circ( B_p, \mathbf{p} + B_p, \ldots \mathbf{(p-1) p} + B_p)		\\
	B_p^{\circ k}	& = circ( B_p^{\circ k-1 }, \mathbf{p^{k-1}} + B_p^{\circ k-1 }, \ldots \mathbf{(p-1) p^{k-1}} + B_p^{\circ k-1})
\end{align*}
with the $circ$ function defined column-wise following \cite{Davis_94}, and $B_p^{\circ k}$ used as short-hand to indicate a matrix therein defined as belonging to the class $\mathcal{BCCB}(p, p, \ldots p)$, $k$ times.

As shown in \cite{Davis_94}, matrices with this structure are diagonalised by $\mathbb{F}_p^{\otimes k}$. This naturally extends the known structure for binary random variables, since as discussed in Section \ref{sec_Walsh_TVD} the convolution matrix for variables in $\left( \mathbb{F}_2 \right)^k$ is diagonalised by the Hadamard matrix $H_{2^k}$, which by construction is equal to $H_2^{\otimes k}$.

The Fourier matrix of size $p$ can be written as a Vandermonde matrix of a primitive $p$-th root of unity as
\begin{align*}
	F_p = 	\begin{pmatrix}
			\omega_p^{0\cdot 0}		& \omega_p^{0\cdot 1}		& \ldots	& \omega_p^{0\cdot (p-1)}	\\
			\omega_p^{1\cdot 0}		& \omega_p^{1\cdot 1}		& \ldots	& \omega_p^{0\cdot 0}		\\
			\ldots					& \ldots					& \ldots	& \ldots			\\
			\omega_p^{(p-1)\cdot 0}	& \omega_p^{(p-1)\cdot 1}	& \ldots	& \omega_p^{(p-1)^2}		
		\end{pmatrix}
\end{align*}
In other words, the entry in row $r$, column $s$ is
\begin{align*}
	F_p(r,s) =	\omega_p^{rs}\,,	\qquad r, s \in \mathbb{Z}_p
\end{align*}
By definition of the Kronecker product of two $(p \times p)$ matrices,
\begin{align}
	K			& = M_1 \otimes M_2	\nonumber \\
	K(u,v)			& = M_1(r_1, s_1) M_2(r_2, s_2)		\quad u, v, r_i, s_i \in \mathbb{Z}_p \nonumber \\
	u			& \equiv r_1 p + r_2	\label{eqn_kron_index_u} \\
	v			& \equiv s_1 p + s_2	\label{eqn_kron_index_v} \\
	(F_p \otimes F_p) (u,v)	& = \omega_p^{r_1 s_1}\omega_p^{r_2 s_2}	\nonumber 
\end{align}
which extends to the $k$-fold Kronecker product by induction using the $p$-ary representation of integers
\begin{align*}
	F_p^{\otimes k}(u,v) = \omega_p^{\mathbf{r}\cdot \mathbf{s}}
\end{align*}
for the specific $\mathbf{r}, \mathbf{s}$ satisfying a polynomial in $p$ such as (\ref{eqn_kron_index_u}) and (\ref{eqn_kron_index_v}) of degree $k-1$. In general, keeping either the row or column index fixed and iterating over the other means iterating over every element of $\left(\mathbb{F}_p\right)^k$; concretely, when evaluating the eigenvalues of a probability mass $\mu \in \mathbb{R}^{p^k}$, the $b$-th eigenvalue corresponds to
\begin{align}
	\lambda(b)	& = F_p^{\otimes k}(b, \cdot) \mu	\nonumber \\
			& = \sum_{j=0}^{p^k-1} F_p^{\otimes k}(b, j) \mu(j)	\nonumber \\
			& = \sum_{j=0}^{p^k-1} \omega_p^{\mathbf{b} \cdot \mathbf{j}} \mu(j)	\,.	\label{eqn_eigs_and_vectors_F_p_k_01}
\end{align}
Generalising from the case of the Walsh transform, this suggests the definition of the $a$-th order Fourier function evaluated at $b$ as
\begin{align*}
	f_a(b)		& = \omega_p^{ \mathbf{b} \cdot \mathbf{a} }
\end{align*}
(compare to Equation (\ref{eqn_Walsh_fn_order_a})), so that if $Y\in \left(\mathbb{F}_p\right)^k$ is the random variable with mass function $\mu$, the eigenvalues may be written as
\begin{align}	\label{eqn_def_phi_j_Fourier_EV}
	\lambda_Y(b) = \mathbb{E}\left[ f_b(Y) \right]	\,.
\end{align}
(compare to Equation (\ref{eqn_def_chi_j_Walsh_EV})).

\begin{lemma}	\label{lemma_ell2_Y_F_p_k}

	The probability mass function $\mu_Y$ of a random variable $Y \in (\mathbb{F}_p)^k$ satisfies
	\begin{align*}
		\left\| \mu_Y - \mu_U \right\|_2	& = \frac{1}{p^{k/2}} \| \lambda_Y - \lambda_U \|_2
	\end{align*}

\end{lemma}
\begin{proof}

	\begin{align}
		\left\| \mu_Y - \mu_U \right\|_2	& = \left\|  \frac{F_p^{\otimes k}}{p^{k/2}}  \left( \mu_Y - \mu_U  \right) \right\|_2	\label{eqn_ell2_S_n_F_q_step_01}	\\
						& = \frac{1}{p^{k/2}} \left\| \lambda_Y - \lambda_U \right\|_2	\label{eqn_ell2_S_n_F_q_step_02}
	\end{align}

\end{proof}
\begin{corollary}
	
	If the elements $Y(j) \in \mathbb{F}_p$ are independent but not necessarily identically distributed,
	\begin{align*}
		\left\| \mu_Y - \mu_U \right\|_2	& = \frac{1}{p^{k/2}} \left\| \bigotimes_{j=0}^{k-1} \lambda_{Y(j)} - \lambda_U \right\|_2
	\end{align*}

\end{corollary}
\begin{proof}

	Since the $X(j)$ are independent, the probability mass function of $Y \in \left(\mathbb{F}_p\right)^k$ is
	\begin{align*}
		\mu_Y	& = \mu_{Y(0)} \otimes \mu_{Y(1)} \otimes \ldots \mu_{Y(k-1)}	\\
			& = \bigotimes_{j=0}^{k-1} \mu_{Y(j)}	\,,
	\end{align*}
	and the eigenvalues will be
	%
	%\begin{align*}
	%	\lambda_X	& = \mathbb{F}_p^{\otimes mn} \mu_X	\\
	%			& = \bigotimes_{j=0}^{n-1} \mathbb{F}_p^{\otimes m} \mu_{X(j)}	\\
	%			& = \bigotimes_{j=0}^{n-1} \lambda_{X(j)}
	%\end{align*}
	%
	%
	\begin{align*}
		\lambda_Y	& = F_p^{\otimes n} \mu_Y	\\
				& = \bigotimes_{j=0}^{k-1} F_p \mu_{Y(j)}	\\
				& = \bigotimes_{j=0}^{k-1} \lambda_{Y(j)}
	\end{align*}
	where the second step follows by the mixed-product property of the Kronecker product.

\end{proof}

We can now extend Theorem \ref{thm_Hadamard_bound_Y} as follows.
\begin{theorem}	\label{thm_Fourier_bound_Y}
	
	The total variation distance of a random vector $Y \in (\mathbb{F}_p)^k$ from uniform,
	\begin{align*}
		\mathrm{TVD}(Y, \mathcal{U})	& = \frac{1}{2} \left\| \mu_Y - \mu_U \right\|_1	\\
					& = \frac{1}{2} \delta_Y
	\end{align*}
	is bounded by
	\begin{align}
		\delta_Y	& \leq \sum_{\mathbf{b} \in (\mathbb{F}_p)^k \setminus \mathbf{0}} \left| \prod_{u = 0}^{k-1} \lambda_{Y}(\mathbf{b}(u)) \right|	\,.	\label{eqn_delta_Y_F_p}
	\end{align}

\end{theorem}
\begin{proof}

	Each eigenvalue may be written as
	\begin{align*}
		\lambda_Y(b)	& = \prod_{u=0}^{k-1} \lambda_Y(\mathbf{b}(u))	\,.
	\end{align*}

	The result follows directly from Lemma \ref{lemma_ell2_Y_F_p_k} and the known bound on vector spaces.

\end{proof}

We can also extend Corollary \ref{crl_Hadamard_bound_Y} to establish a connection with the number of vectors of a specific Hamming weight, but in the non-binary case we can also go into more detail if the full composition of each vector in the space is known, as in the following definition.

\begin{definition}	\label{def_weight_enumerator}
	Let $s(\mathbf{b})$ be the composition of $\mathbf{b} \in (\mathbb{F}_p)^k$ such that $s_j(\mathbf{b})$ is the number of components of $\mathbf{b}$ equal to $j$.%, and let $S(\mathbf{t})$ be the number of $\mathbf{b} \in (\mathbf{F}_p)^k$ with composition $\mathbf{s}$:
	\begin{align*}
		%\mathbf{b}	& \in (\mathbb{F}_p)^k	\\
		s(\mathbf{b})	& = (s_0, s_1, \ldots s_{p-1})	\\
		s_j		& = \left| \set{i | \mathbf{b}(i) = j} \right|
	\end{align*}
	Let $W_{(\mathbb{F}_p)^k}(t)$ be the enumerator of the elements $\mathbf{b}$ having composition equal to $t$, with $t$ being a $p$-tuple summing to $k$:
	\begin{align*}
		W_{(\mathbb{F}_p)^k}(t)	& = |\set{\mathbf{b} \in (\mathbb{F}_p)^k  | s(\mathbf{b}) = t}|	\\
		t			& \in T \subset \mathbb{N}^p	\\
		\sum_j t_j		& = k	\,;
	\end{align*}
	then the number of $\mathbf{b}$ with Hamming weight equal to $l$ is
	\begin{align*}
		A_l	& = \sum_t W(t)	\qquad t \in \set{ t_0 = k-l }.
	\end{align*}
	In particular, if instead of $\mathbf{b} \in (\mathbb{F}_p)^k$ we consider a set of codewords $\mathbf{c} \in C$, the enumerator $W_{C}$ is the complete weight enumerator of $C$, and $A_l$ its weight distribution, as defined in \cite{MacWilliams_77} ch.\ $5\S6$.
\end{definition}

\begin{corollary}	\label{cor_Fourier_bound_Y}

	If each $Y(j)$ is i.i.d., the total variation distance of a random vector $Y \in (\mathbb{F}_p)^k$ from uniform is bounded by
	\begin{align}
		\delta_Y	& \leq \sum_t W(t) \prod_{u=0}^{p-1} \left(\lambda_{Y(j)}(u)\right)^{t_u}	\qquad t \in \set{ t_0 < k }	\label{eqn_delta_Y_F_p_CWE}
	\end{align}
		
	Without knowledge of the full spectrum of $Y(j)$ one may obtain a coarser bound using the second largest eigenvalue is $\lambda_{Y*}$, as in Definition \ref{defn_lambda_star}:
	\begin{align}
		\delta_Y	& \leq  \sum_{l=1}^k A_l \lambda_{Y*}^l	\,,	\label{eqn_delta_Y_F_p_WD}
	\end{align}
	where $A_l$ is the number of $\mathbf{b} \in (\mathbb{F}_p)^k$ with Hamming weight $w(\mathbf{b}) = l$.

\end{corollary}
\begin{proof}

	Each eigenvalue may further be written as
	\begin{align*}
		\lambda_Y(b)	& = \prod_{u=0}^{k-1} \lambda_Y(\mathbf{b}(u))	\\
				& = \prod_{u=0}^{k-1} \sum_{v=0}^{p-1} \omega_p^{\mathbf{b}(u) \cdot v} \mu_{Y(u)}(v)	\,,
	\end{align*}
	so all the $\mathbf{b}$ with identical composition $t$ will correspond to equal eigenvalues, leading directly to Equation (\ref{eqn_delta_Y_F_p_CWE}). If the Hamming weight $w(\mathbf{b}(u)) = 0$, then the $u$-th term of the product will be equal to $1$; Equation (\ref{eqn_delta_Y_F_p_WD}) follows by considering the worst case $\lambda_Y(j) = \lambda_{Y*} \; \forall j>0$.

\end{proof}

\label{sec_Fourier_TVD_single_vector}

\section{Fourier bound on entropy extractors}
\label{sec_Fourier_extractor}

In order to arrive at a bound involving the distribution of weights, we begin by showing there is an unique association between code words and eigenvalues, just as there was with the bias of individual bits in the binary case (see Theorem \ref{thm_Hadamard_bound_GX}).

If $Y = GX$, with $X$ a random vector in $(\mathbb{F}_p)^n$, $G$ a generator matrix of an $(n,k,d)$ code over $\mathbb{F}_p$, we can establish a direct link between eigenvalues of $Y$ and codewords of $G$ using Equation (\ref{eqn_def_phi_j_Fourier_EV}):
\begin{align}
	\lambda_Y(b)	& = \mathbb{E}\left[ f_b(GX) \right]	\nonumber \\
			& = \sum_{j=0}^{p^n-1} \omega_p^{\mathbf{b}^T G \mathbf{j}} \mu_X(j)	\nonumber \\
			& = \sum_{j=0}^{p^n-1} \omega_p^{\mathbf{c} \cdot \mathbf{j}} \mu_X(j)	\label{eqn_eig_b_F_p}
\end{align}
with $\mathbf{c} = \mathbf{b}^T G$ a particular word of the code. Note that choosing a particular $(k \times n )$ matrix $G$ is equivalent to selecting the specific $p^k$ rows specified by all the $k$ codewords $\mathbf{c}$ that forms a subset of the $p^n$ rows of $F_p^{\otimes n}$ by which to multiply $\mu_X$.

Having noted this fundamental link in principle in the same manner as for the binary case (see Equation (\ref{eqn_chi_and_bias_GX})), and having developed the required tools in Section \ref{sec_Fourier_TVD_single_vector}, we can immediately state some more specific results for particular cases of practical interest, beginning with an extension of Theorem \ref{thm_Fourier_bound_Y}.

\begin{theorem}

	Let $Y = GX$, where $X \in \left(\mathbb{F}_p\right)^n$ is a random vector of length $n$, with each entry being an independent but not necessarily identically distributed variable $X(j) \in \mathbb{F}_p$ with mass function $\mu_{X(j)} \in \mathbb{R}^p$, and $G$ is the generator matrix of an $(n,k,d)$ linear code over $\mathbb{F}_p$. Then the $b$-th eigenvalue of the distribution of $Y$ is
	\begin{align}
		\lambda_Y (b) = \prod_{j=0}^{n-1} \lambda_{X(j)}(\mathbf{c}(j))	\label{eqn_lambda_Y_b}
	\end{align}
	where $\mathbf{c}(j) \in \mathbb{F}_p$ is the $j$-th symbol in the codeword $\mathbf{c}^T = \mathbf{b}^T G$.
	
\end{theorem}
\begin{proof}

	The specific combination corresponding to a word $\mathbf{c}$ is from Equation (\ref{eqn_eig_b_F_p}):
	\begin{align*}
		\lambda_Y(b)	& = \sum_{j=0}^{p^n-1} \omega_p^{\mathbf{c} \cdot \mathbf{j}} \mu_{X(j)}	\\
				& = \prod_{u=0}^{n-1}\sum_{v=0}^{p-1} \omega_p^{\mathbf{c}(u) \cdot \mathbf{v}} \mu_{X(u)}(v)	\,.
	\end{align*}

\end{proof}

\begin{corollary}	\label{cor_lambda_Y_b_w_c}

	If all $X(j)$ are also i.i.d., the total variation distance of a random vector $Y \in (\mathbb{F}_p)^k$ from uniform is bounded by
	\begin{align}
		\delta_Y	& \leq \sum_t W_{C}(t) \prod_{u=0}^{p-1} \left(\lambda_{X(j)}(u)\right)^{t_u}	\qquad t \in \set{ t_0 < n }	\label{eqn_delta_X_F_p_CWE}
	\end{align}

	Without knowledge of the full spectrum of $X(j)$ one may obtain a coarser bound using the second largest eigenvalue is $\lambda_{X*}$, as in Definition \ref{defn_lambda_star}:
	\begin{align}
		\delta_Y	& \leq  \sum_{l=d}^n A_l \lambda_{X*}^l	\,,	\label{eqn_delta_X_F_p_WD}
	\end{align}
	Here $W_{C}$ and $A_l$ are the complete weight enumerator and weight distribution of $C$, respectively, as in Definition \ref{def_weight_enumerator}.

\end{corollary}
\begin{proof}

	The proof follows in the same manner as for Corollary \ref{cor_Fourier_bound_Y}.
	% If the probability mass functions $\mu_{X(j)}$ are identical then so are the eigenvalues, and for any valid probability mass function the zeroth eigenvalue will be equal to $1$; hence, any symbol of $\mathbf{c}$ that is equal to zero will lead to a multiplication by $1$ in Equation (\ref{eqn_lambda_Y_b}).

\end{proof}

The above can be viewed as a statement regarding the sum of $n$ random variables, each in $\mathbb{F}_p$: if only $w(\mathbf{c})$ symbols are non-zero, this corresponds to a sum of $w(\mathbf{c})$ terms.

\begin{corollary}	\label{cor_Lacharme_extended}
		
	Using the minimum distance $d$, one may obtain the bound
	\begin{align*}
		\delta_Y	& \leq  (p^k-1) \lambda_{X*}^d	\,.
	\end{align*}

\end{corollary}

Note that all the results in this section extend to random vectors $X \in \left( \mathbb{F}_{p^m} \right)^n$, that is to sequences of random vectors taking values in $\mathbb{F}_{p^m}$ by using the right matrix to diagonalise the convolution matrix of the sum of two such variables in order to compute the eigenvalues, and assuming the symbols of the generator matrix are taken in the same field, i.e.\ the code is chosen over $\mathbb{F}_{p^m}$. Following Section \ref{sec_Fourier_TVD_single_vector}, this may be done using the Kronecker product $F_p^{\otimes m n}$.

Comparing Corollaries \ref{cor_Lacharme} and \ref{cor_Lacharme_extended}, it appears that a bound based solely on the minimum distance quickly risks becoming far from sharp as the dimension of the underlying random variable $X(j)$ increases.

\section{Non-linear codes}

As shown in \cite{Lacharme_08}, it is possible to construct ad-hoc non-linear maps with better properties than linear ones for specific cases; it was also noted that for a given compression ratio $k/n$ of the output, there may exist non-linear codes with a greater minimum distance than any linear code. Since non-linear codes do not have a generator matrix $G$ they are not straightforward to cover using the tools developed thus far, but we may use some of them to frame the fundamental issue with non-linear maps, as we see it, in terms of examining the distribution of the product of random variables. Consider the special case $X_1, X_0 \in \mathbb{F}_2$, and let their product be $P_2$; its mass function may be written as follows:
\begin{align*}
	P_2		& = X_1 X_0	\\
	\mu_{P_2}	& =	\begin{pmatrix}
					1	& \mu_{X_1}(0)	\\
					0	& \mu_{X_1}(1)
				\end{pmatrix}
				\begin{pmatrix}
					\mu_{X_0}(0)	\\
					\mu_{X_0}(1)
				\end{pmatrix}
\end{align*}
As long as neither $X_0$ or $X_1$ follow the categorical distribution with $\mathbb{P}(1)=1$, the probability of their product being zero is strictly greater than either of the initial probabilities. By induction,
\begin{align*}
	\lim_{j \rightarrow \infty} \mu_{P_j} = \begin{pmatrix}	1	\\ 0	\end{pmatrix}.
\end{align*}
We may conclude that, while increasing the number of linear operations will lead to the uniform distribution in expectation, increasing the number of non-linear operations will in general lead to a worsening of the output distribution, except in very specific cases. By way of example, consider the two Bernoulli variables
\begin{align*}
	B_+		& \sim \mathcal{B}(2^{-1/2}), \quad	& B_- \sim \mathcal{B}(1-2^{-1/2})	\\
	\mu_{B_+}	& = \begin{pmatrix}
				1- 2^{-1/2}	\\
				2^{-1/2}
			\end{pmatrix}
			\quad	&
				\mu_{B_-} = \begin{pmatrix}
						2^{-1/2}	\\
						1-2^{-1/2}
					\end{pmatrix}.
\end{align*}
Note that the bias of these two random variables is identical; however, the distributions of their products are quite different:
\begin{align*}
	\mu_{B_+ B_+}	& = \begin{pmatrix}
				2^{-1}	\\
				2^{-1}
			\end{pmatrix}
			\quad	&
				\mu_{B_- B_-} = \begin{pmatrix}
						2^{1/2} - \frac{1}{2}	\\
						\frac{3}{2} - 2^{1/2}
					\end{pmatrix}.
\end{align*}
%
%While this processing does give an optimal result in the case of the product of $B_+$ with itself, the distribution of any further products with bits $B_+$ from the same generator will produce results that are progressively further from uniform
%
While it is possible to find non-linear maps that are optimal in some specific cases, we observe that not only does repeated processing by nonlinear maps in general lead to a worsening of the output, but it is also necessary to know or assume a specific distribution of the sequence to be processed to even attempt to find such a processing; even under the assumption of i.i.d.\ binary variables, knowledge of the bias of each bit is not sufficient.

\section{Conclusions}

We have shown new bounds on the statistical distance from the uniform distribution of random number sequences conditioned by linear transformations chosen from the generator matrices of linear codes, based on the assumption of independent generator output in $\mathbb{F}_{p^m}$; we have also shown how these bounds are natural generalisations of known bounds in $\mathbb{F}_{2^m}$ once the structure behind the known bounds is made clear. If the weight distribution or the complete weight enumerator of the code is known, this allows one to determine the distribution of the conditioned output exactly. This is of especial importance whenever a matrix is chosen once, possibly based on a random seed, and then seldom changed, if ever. %When appropriate, these bounds allow the practitioner both to assess the performance of a chosen matrix as well as to make an informed choice from a pre-existing set with well-defined properties. This would seem particularly advantageous with respect to randomly choosing boolean matrices as conditioners and relying on the leftover hash lemma \cite{Impagliazzo_Levin_Luby_89} to conclude sufficiently good output will be produced in expectation. 

\section*{Acknowledgments}

The authors would like to thank the anonymous referee of \cite{Meneghetti_Sala_Tomasi_14} for suggesting an approach to the proof of Theorem \ref{thm_Hadamard_bound_GX}.

The second author would like to thank his PhD supervisor (the third author).

This research was funded by the Autonomous Province of Trento, Call ``Grandi Progetti 2012'', project ``On silicon quantum optics for quantum computing and secure communications - SiQuro''.

%\printbibliography

\bibliographystyle{elsarticle-num}
\bibliography{../RefsAAlgebra,../RefsACoding,../RefsARNG,../RefsAStochastics,../RefsAWalsh,../RefsA}

\end{document}